\documentclass[notitlepage]{article}
\usepackage[margin=1.0in]{geometry}
\usepackage{titling}
\usepackage{enumerate}   
\usepackage{wrapfig}    
\usepackage{graphicx}
\usepackage{subfloat}
\usepackage{subfigure}
\usepackage{multirow}
\usepackage{array}
\usepackage{algorithm}
\usepackage{algpseudocode}
\usepackage[compatibility=false]{caption}
\usepackage{algcompatible,amssymb,amsmath}
\usepackage{caption}  
\usepackage{comment} 

\newtheorem{theorem}{Theorem}[section]
\newtheorem{proposition}[theorem]{Proposition}

\newtheorem{lemma}[theorem]{Lemma}
 
\newtheorem{definition}[theorem]{Definition}
\newenvironment{proof}{\paragraph{Proof:}}{\hfill$\square$}
 
\title{Optimal Load Balanced Demand Distribution under Overload Penalties} 
\author{Sarnath Ramnath \\St. Cloud State University, USA \and Venkata M. V. Gunturi \\IIT Ropar, India}
\date{}

\begin{document}
\maketitle
   
    \begin{abstract}
        Input to the Load Balanced Demand Distribution (LBDD) consists of the following: (a) a set of  service centers; (b) a set of demand  nodes and; (c) a cost matrix containing the cost of assignment for each (demand node, service center) pair. In addition, each service center is also associated with a notion of capacity and a penalty which is incurred if it gets overloaded. Given the input, the LBDD problem determines a mapping from the set of $n$ demand vertices to the set of $k$ service centers, $n \gg k$. The objective is to determine a mapping that minimizes the sum of the following two terms: (i) the total cost between demand units and their allotted service centers and, (ii) total penalties incurred. The problem of LBDD finds its application in a variety of applications. An instance of the LBDD problem can be reduced to an instance of the min-cost bi-partite matching problem. The best known algorithm for   min-cost matching in an unbalanced bipartite graph yields a complexity of $O(n^3k)$.   This paper proposes novel allotment subspace re-adjustment based approach   which allows us to characterize the optimality of the mapping without invoking matching or mincost flow. This approach yields an optimal solution with time complexity $O(nk^3 + nk^2 \log n)$, and also allows us to efficiently  maintain an optimal allotment under insertions and deletions. 
    \end{abstract}

\section{Introduction}
\label{intro}
The problem of Load Balanced Demand Distribution (LBDD) takes the following three as input. (a) a set $S$ of service centers (e.g., COVID clinics, schools etc.); (b) a set $D$ of demand units (e.g., people); (c) a cost matrix $\mathcal{CM}$ which contains the cost of assigning a demand unit $d_i\in D$ to a service center $s_j \in S$ (for all $<d_i,s_j>$ pairs). Additionally, each service center $s_j \in S$ is associated with a positive integer capacity and a notion of ``penalty'' which denotes the \textit{"extra cost"} that must be paid to overload the particular service center. Given the input, the LBDD problem determines an mapping between the set of demand units and the set of service centers. The objective here is to determine a mapping which minimizes the sum of the following two terms: (1) total cost of assignment (accumulated across all assignments) and, (2) total penalty incurred (if any) while overloading the service centers.

\noindent \textbf{Problem Motivation:} 
The Load Balanced Demand Distribution (LBDD) problem finds its application in the domain of urban planning. For instance, consider the task of defining the geographic zones of operation of service centers such as schools (refer \cite{Catchment_area}), COVID clinics and walk-in COVID-19 testing centers (in case of continuous monitoring of the disease at a city scale). 

The key aspect over here being that each of the previously mentioned type of service centers is associated with a general notion of capacity. This capacity dictates the number of people (demand) that can be accommodated (comfortably) each day, week or during any specific duration of time (e.g., typical duration of sickness of patients). In addition, the quality of service at any of these service centers is expected to degrade if significantly more number of people (beyond its capacity) are assigned to it. The cost of this degradation can be modeled as a penalty function associated with each service center.

\noindent

\noindent \textbf{Limitations of Related Work}

The current state of the art most relevant to our work includes the work done in the area of network voronoi diagrams without capacities \cite{okabe,Demiryurek2012}, network voronoi diagrams under capacity constraints \cite{7123646,U2010,U2008}, weighted voronoi diagrams \cite{pdiag} and optimal location queries (e.g., \cite{Yao2014,Xiao2011,MelkoteD01,Diabat16}) 

Work done in the area of network voronoi diagrams without capacities \cite{okabe,Demiryurek2012} assume that the service centers have infinite capacity, an assumption not suitable in many real-world scenarios. On the other hand, work done in the area of network voronoi diagrams with capacities \cite{7123646,U2010,U2008,Yang2013} did not consider the notion of ``overload penalty.'' They perform allotments (of demand units) in an iterative fashion as long as there exists a service center with available capacity. In other words, the allotments stop when all the service centers are full (in terms of their capacity). This paper considers the problem in a more general setting in the sense that we allow the allotments to go beyond the capacities of the service centers. And after a service center is full, we use the concept of the \emph{overload penalties} for guiding the further allotments.

Weighted voronoi diagrams \cite{pdiag} are specialized voronoi diagrams. In these diagrams, the cost of allotting a demand unit $x$ to a service center $p$ is a linear function of the following two terms: (i) distance between $x$ and $p$ and, (ii) a \emph{real number} denoting the weight of $p$ as $w(p)$. LBDD problem is different from weighted voronoi diagrams. Unlike the weighted voronoi diagrams, our ``$w(p)$'' is a \emph{function} of the number of allotments already made to the service center $p$. And it would return a non-zero value only when the allotments cross beyond the capacity. Whereas in \cite{pdiag}, $w(p)$ is assumed to play its role throughout. 
 
Optimal location queries (e.g., \cite{Yao2014,MelkoteD01,Diabat16}) focus on determining a new location to start a new facility while optimizing a certain objective function (e.g., total distance between clients and facilities). Whereas, in LBDD, we already have a set of facilities which are up and running, and we want to load balance the demand around them.

The LBDD problem can be theoretically reduced (details in \cite{NET20477}) to min-cost matching in an unbalanced bi-partite graph, which yields a complexity of $O(n^3k + n^2\log n)$ when a generalized version of the Hungarian algorithm is used \cite{tarjan2012}.

\noindent \textbf{Our  Contributions:} This paper makes the following contributions:

\noindent (a) We define the concept of an allotment multigraph that captures all the ways in which an allotment can be perturbed, and the cost associated with each perturbation. 

 \noindent (b) We give an alternate characterization of optimality using the concept of negative loops induced by the allotment multigraph.
 
10

\noindent (c) The allotment multigraph treats the demand nodes as entities that can be pushed from one service center to another.  This approach allows us to efficiently handle a dynamic situation where the optimal allotment needs to be maintained under addition and deletion of demand nodes. We show that each dynamic operation can be  performed in  $O(k^3 + k^2logn)$ steps, where $k$ is the number of service centers and $n$ is the number of demand units. 

\noindent (d)   Instead of viewing the problem as that of optimally matching demand nodes to service center nodes, we  treat it as the problem of maintaining an optimal allocation during a sequence of $n$ demand node additions. This  improves the complexity of the problem from $O(n^3k)$ (complexity of min-cost bi-partite matching) to $O(nk^3 + nk^2logn)$.

\noindent \textbf{Outline:} The rest of this paper is organized as follows: In Section \ref{bc}, we provide the basic concepts and the problem statement. Section \ref{pa} presents our proposed approach. Section \ref{optsec} proves the correctness of our algorithms.   Finally in Section \ref{con} we conclude the paper.

\section{Basic Concepts and Problem Definition}
\label{bc}
\begin{definition}
\textbf{A Service Center} is a public service unit of a particular kind (e.g., schools, hospitals, COVID-19 testing centers in a city). A set of service centers is represented as $S=\{s_1,..., s_{n_s}\}$.  
\end{definition}

\begin{definition}
\textbf{A Demand unit} represents a unit population which is interested in accessing the previously defined service center. A set of demand units is represented as  $D=\{{d_1},..., d_{n_d}\}$, $n_d$ is the number of demand units. 
\end{definition}

\begin{definition}
\textbf{Capacity of a service center} $(c_{i})$ is the prescribed number of  demand units that a service center $s_i$ can accommodate comfortably. For example, the number of people a COVID-19 testing facility can test in a day (or a week) would define its capacity.  Similarly, the number of students a school can admit would correspond to the notion of capacity defined in this paper.
\end{definition}

\begin{definition}
\label{pdef}
\textbf{Penalty function for a service center} $(q_{i}())$ is a function which returns the ``extra cost'' ($>0$) that must be paid for every new allotment to the service center $s_i$ which has already 
  exhausted its capacity $c_{i}$.  
Penalty function takes into account the current status of $s_i$ (i.e., how many nodes have been already added to $s_i$) and then returns a penalty for the $j^{th}$ ($1\leq j \leq (n_d-c_{s_i})$) 
allotment. $q_{i}()$ returns only positive non-zero values and is monotonically increasing over $j$ ($1\leq j \leq (n_s-c_{i})$).  
\end{definition}
 
\begin{definition} 
\textbf{Demand-Service Cost Matrix $\mathcal{CM}$:} contains the cost (a positive integer) of assigning a demand unit $d_i\in D$ to a service center $s_j \in S$ (for all $<d_i,s_j>$ pairs). If in the given LBDD problem instance, service centers and demand units, come from a Geo-Spatial reference frame then the cost of assignment a demand unit $d_i$ to a service center $s_j$ can represent things such as shortest distance over the road network, travel-time and/or cost of traveling from $d_i$ to $s_j$ along the shortest path (in Dollars), Geodetic distance or Euclidean distance. 
\end{definition}

\subsection{Problem Statement}
\label{probdef}
We now formally define the problem of load balanced demand distribution by detailing the input, output and the objective function:

\noindent \textbf{Given}:			
			\begin{itemize}
			\item A set of service centers $S=\{s_1,..., s_{n_s}\}$.
			\item A set of demand units  $D=\{{d_1},..., d_{n_d}\}$.
			\item A demand-service cost matrix $\mathcal{CM}$ which contains the cost of allotting a demand unit $d_i \in D$ to a service center $s_j \in S$ ($\forall d_i \in D,~\forall s_j \in S$). 
			\item Capacity $(c_{i})$ of each service center $s_i \in S$.  
			\item Penalty function $(q_{i}())$ of each service center $s_i \in S$.
			\end{itemize}
			
\noindent \textbf{Output}: A mapping from the set of demand units to the set of service centers.  Each demand unit is allotted to only one service center. 
			
\noindent \textbf{Objective Function}: 
			\begin{multline}
			 \textit{\textbf{Minimize}}
			 \Bigg\{
			 \sum_{\substack{s_i \in Service \\Centers}} 	
				\Bigg\{
				\sum_{\substack{d_j \in Demand  \ unit\\ allotted \ to \  s_i}} 
					\mathcal{CM}(d_j,s_i)
					\Bigg\} + Total~Penalty~across~all~s_i \Bigg\}
			\label{eq2}   
			\end{multline}

\subsection{Variations of the LBDD problem} 
The LBDD problem can be varied along following two dimensions: (a) relationship between the total capacity (across all service centers) and the total demand; (b) Presence or absence of penalty functions on service centers. 
Along the first dimension, total capacity can be less, greater or more than the total capacity of the service centers.  On the second dimension, we have following two cases: (i) \emph{any service center} can be overloaded and, (ii) \emph{no service center} can be overloaded. Note that according to our definition, a service center $s_i$ can be overloaded beyond its capacity only when its corresponding penalty function ($q_{i}()$) is defined. Otherwise, $s_i$ must not be assigned more than $c_{i}$ demand units. 

Note that when total demand is greater than the total capacity and service centers are not allowed to be overloaded, only some demand units would be assigned to service centers. In this case, our algorithm would implicitly pick the optimal set of demand nodes which need to assigned to service centers such that the objective function attains its lowest value.  
 
 The following theorem establishes that the big-oh complexity of the case where overloading is not permitted is at most that of case where penalties are assessed for overloading.
 
 \begin{theorem}

\label{subsumeThm}
Let $\mathcal{L}$ be an instance of LBDD with $n_d$ demand node, $n_s$ service center nodes, and a cost matrix $\mathcal{CM}$, with the constraint that no service center can be overloaded.
 We can construct 
an instance $\mathcal{L}_1$ of LBDD, consisting of  a cost matrix $\mathcal{CM}_1$, that has $n_d$ demand nodes and $n_s+1$ service center nodes such that:\\
 (1) $\mathcal{L}_1$ allows us to exceed the  capacity for all service centers, and has a penalty function.\\
(2) Any optimal allocation, $\mathcal{A}_1$, for $\mathcal{L}_1$,   will contain within it an optimal allocation, $\mathcal{A} $, for $\mathcal{L} $. \\
(3) Let $\mathcal{A}_1$  be any  allocation   for $\mathcal{L}_1$ such that $\Gamma(\mathcal{A}_1)$ does not have any negative cost cycles. $\mathcal{A}_1$ contains within it an  allocation $\mathcal{A} $, for $\mathcal{L}$, such that $\Gamma(\mathcal{A})$ does not have any negative cost cycles.
 
\end{theorem} 
\begin{proof}
 Let $max$ be the largest value in $\mathcal{CM}$. We construct $\mathcal{CM}_1$ by adding a  column for service center $s_{n_s+1}$ and set all the values in this row to $max+1$.  To define $\mathcal{L}_1$, we set the penalty $q(s_i) =  \infty$ for $1 \leq i \leq n_s$, and $q(s_{n_s+1}) = 0$.
 Let $\mathcal{A}_1$  be an optimal  allocation   for $\mathcal{L}_1$. From the values in $\mathcal{CM}_1$ and the penalty function,  it follows that $\mathcal{A}_1$ must  fill service centers $s_1, s_2,\ldots s_{n_s} $ to capacity, and assign the remainder to  $s_{n_s+1}$. Since all the demand nodes have the same cost for $s_{n_s+1}$, it follows that  $\mathcal{A}_1$ should choose the best possible subset of demand nodes for $s_1, s_2,\ldots s_{n_s} $.  Therefore, if $\mathcal{A}_1$ is optimal, the allotment $\mathcal{A}$, formed by restricting 
 $\mathcal{A}_1$ to service centers $s_1, s_2,\ldots s_{n_s} $ is an optimal assignment for $\mathcal{L}$. Since $\Gamma(\mathcal{A})$ is a subgraph of $\Gamma(\mathcal{A}_1)$, the third item follows.   
 \end{proof}

\section{Proposed Approach}
\label{pa}
This section presents our proposed Allotment Subspace Re-adjustment based approach   for the LBDD problem. With the intention of keeping the discussion concise, we mainly focus only the variant where the total capacity is less than total demand, and the service centers are allowed to be overloaded (after ``paying'' their respective penalty). The remaining of section is organised as follows. Section \ref{over} presents an overview of our proposed approach. In Section \ref{asr}, we introduce our key computational idea of allotment subspace re-adjustment.

\subsection{Overview of the algorithm}
\label{over}
Overall, our algorithm follows an incremental strategy to build the optimal solution. Say that we have an optimal allotment for $i-1$ demand nodes. In the next stage we add another demand node $d_i$ to some service center $s_j$.

Now, we attempt to adjust our current allotment to find the new optimal allocation. The adjustment involves identifying negative cycles and negative paths in   what we refer to as the \emph{allotment subspace multigraph} (details in next section). 
This identification, as we shall see, can be done using the Bellman-Ford shortest path algorithm. Following this, we remove these negative cycles and paths by shifting 
demand nodes along these cycles and paths and finding  a new allotment.

\subsection{Allotment Subspace Re-adjustment} 
\label{asr}
The key aspect of this paper is the concept of re-adjustments in the space of, what we refer to as, the \emph{allotment subspace multi-graph}. Each service center (given in the problem instance) forms a node in this allotment subspace multi-graph. Each edge in the multigraph represents the potential transfer of a demand node from one service center another. Formally, we define this as follows. 

\begin{definition}
Given an allotment $\mathcal{A}$ comprising of tuples of the form $<d_i,s_j>$ (demand unit $d_i$ is allotted to $s_j$), the allotment sub-space multi-graph $\Gamma(\mathcal{A}) = (M,N)$ is defined as follows:
\begin{itemize}
    \item \textbf{Set of nodes} ($M$): Each service center $s_i \in S$ is a node in $M$.
    \item  \textbf{Set of edges} ($N$): Each edge is defined as a triple $e=(s_i,s_j,dn)$. Here, $e$ represents the potential transfer of the demand unit $dn$ (which is currently allotted to $s_i$ in $\mathcal{A}$) from $s_i$ to $s_j$. Cost of $e$ is defined as $\mathcal{CM}(dn,s_j)-\mathcal{CM}(dn,s_i)$. This cost is referred to as the \emph{transfer-cost} of demand unit $dn$. These edges are referred to as {\em demand node transfer edges}, or simply {\em transfer edges}, when they need to be disambiguated from other edges.
\end{itemize}
\label{Gammadef} 
\end{definition}

\begin{definition}
Given an allotment $\mathcal{A}$,   {\bf Mincost Edge Graph}, $\Gamma_{min}(\mathcal{A}) = (V, E)$ is a complete directed graph constructed as follows:  
\begin{itemize}
    \item \textbf{Set of nodes} ($V$): Each service center $s_i \in S$ is a node in $V$.
    \item  \textbf{Set of weighted edges} ($E$): Let $D_i$ be the set of demand nodes assigned to service center $s_i$. 
    The weight of the edge, $(s_i, s_j)$, directed from $s_i$ to $s_j$ is  least cost 
    of transferring a demand node from $D_i$ to $s_j$. Formally, it is the smallest number in the set $\{dn \in D_i \mid (\mathcal{CM}(dn, s_j) - \mathcal{CM}(dn, s_i) )\}$.  We refer to this edge as the minimum weight (cost) edge from $s_i$ to $s_j$ in $\Gamma (\mathcal{A})$, denoted $\Gamma_{min} (i,j)$.
\end{itemize}

	\begin{definition}
  Given an allotment $\mathcal{A}$ the {\bf occupancy} of a service center $s_j$, denoted $o_j$,  is defined as the number of demand nodes assigned to $s_j$ under $\mathcal{A}$. The value of the penalty function  $(q_{j}(o_j))$, the amount of penalty incurred when the last demand node was assigned to $s_j$.   $(q_{j}(o_j+1))$ is the penalty that will be incurred when the next demand node is assigned to $s_j$. In the event of ambiguity about the allotment being referred to, we use the notation $o_j(\mathcal{A})$ to specify the allotment $\mathcal{A}$.
\label{occupydef} 
\end{definition}

\label{Mincostgraphdef} 
\end{definition}
\begin{definition}
Given an allotment $\mathcal{A}$,   {\bf Penalty Transfer Graph}, $\Gamma{pen}(\mathcal{A}) = (V, E)$ is a complete directed graph constructed as follows:  
\begin{itemize}
    \item \textbf{Set of nodes} ($V$): Each service center $s_i \in S$ is a node in $V$.
    \item  \textbf{Set of weighted edges} ($E$):   
    The weight of the edge, $e_{ij}$, directed from $s_i$ to $s_j$  is the change in penalty when a demand node currently assigned to $s_j$ is transferred to $s_i$. Formally,  the weight of $(s_i, s_j)$ is $q_{i}(o_i+1)- q_{j}(o_j).$ We refer to this edge as the penalty transfer edge from $s_i$ to $s_j$ in $\Gamma (\mathcal{A})$, denoted $\Gamma_{pen} (i,j)$.
\end{itemize}
\label{Penaltygraphdef} 
\end{definition}

Consider Figure \ref{ncyca} which illustrates a partially constructed solution for a sample problem instance. For ease of understanding, the figure illustrates the LBDD instance where both service centers and demand units are present in a road network represented as directed graph. Here, nodes $S1$, $S2$, $S3$ and $S4$ are service, whereas other nodes (e.g., $A$, $B$, $C$, etc) are demand units (people). In this problem, cost of assigning a demand node $v$ to a service center $r$ is the shortest distance between $v$ and $r$. For instance, cost of assigning demand unit $B$ to $S1$ is the shortest distance between $B$ and $S1$ in the graph (which is $5$ in this case). The figure also details the total capacity and penalty values for each of the service centers. In the partial assignment shown in Figure \ref{ncyca}, the first few demand vertices ($A,B,C,D$ and $E$) have already been processed. Demand vertices which are allotted a service center are filled using the same color as that of their allotted service center. For e.g., demand vertices $A$, $B$ are assigned to resource unit $S1$. Nodes which are not yet allotted are shown without any filling.

The allotment sub-space graph of this network would contain four nodes (one for each service center). And, between any two service centers $s_i$ and $s_j$, it would contain directed edges representing transfer of demand nodes across the service centers. Edges directed towards $s_j$ (from $s_i$) would represent demand vertices being given to $s_j$ (from $s_i$). 

Figure \ref{ncycb} illustrates some edges of the allotment sub-space graph ($\Gamma()$) of the allotment shown in Figure \ref{ncyca}. To maintain clarity, we do not show all the edges in the Figure. Figure \ref{ncycb} illustrates only the edges whose tail node is service center S1. Now, given that S1 was allotted two demand nodes ($A$ an $B$ according to Figure \ref{ncyca}), the node corresponding to S1 in the allotment sub-space graph would have six edges (two edges to each of the other service centers) coming from S1. For instance, consider the two edges directed from S1 to S2 in Figure \ref{ncycb}. One of them represents transfer of demand node $A$ and the other one represents the transfer of demand node $B$ to the service center S2. Cost of the edge is defined as the difference in distance to the service centers. For instance, cost of edge corresponding to $A$ in the graph is defined as $\mathcal{CM}(A,S2)-\mathcal{CM}(A,S1)$ (which is $2$ in our example). Note that edge costs in the allotment sub-space graph may be negative in some cases.

\subsubsection{Maintaining $\Gamma_{min}$}

 We can do this using $2\times$ ${n_s}\choose {2}$ number of minheaps, where $n_s$ is the number of service centers. Basically, one heap for each ordered pair of service centers. This heap, \emph{referred to as BestTransHeap}, would ordered on the transfer cost of demand units (across a pair of service centers). For any order pair of service centers <$s_i,s_j$>, the top of its corresponding $BestTransHeap_{ij}$ would contain the demand unit (which is currently allotted to $s_i$) which has the lowest transfer cost to $s_j$. And while constructing $\Gamma_{min}()$, for any ordered pair <$s_i,s_j$>, we would create only one edge (between $s_i$ and $s_j$) which corresponds to the demand unit at the top of $BestTransHeap_{ij}$

\subsubsection{Using allotment subspace  multigraph for improving the solution}
The allotment subspace  multigraph helps us to improve the current solution in following two ways: (a) Negative cycle removal and, (b) Negative path removal. As one might expect, these two are not completely independent of each other.     

\begin{figure*}[ht] 
\begin{center} 
\subfigure[Sample allotment.]{\label{ncyca}\includegraphics[width=0.39\textwidth]{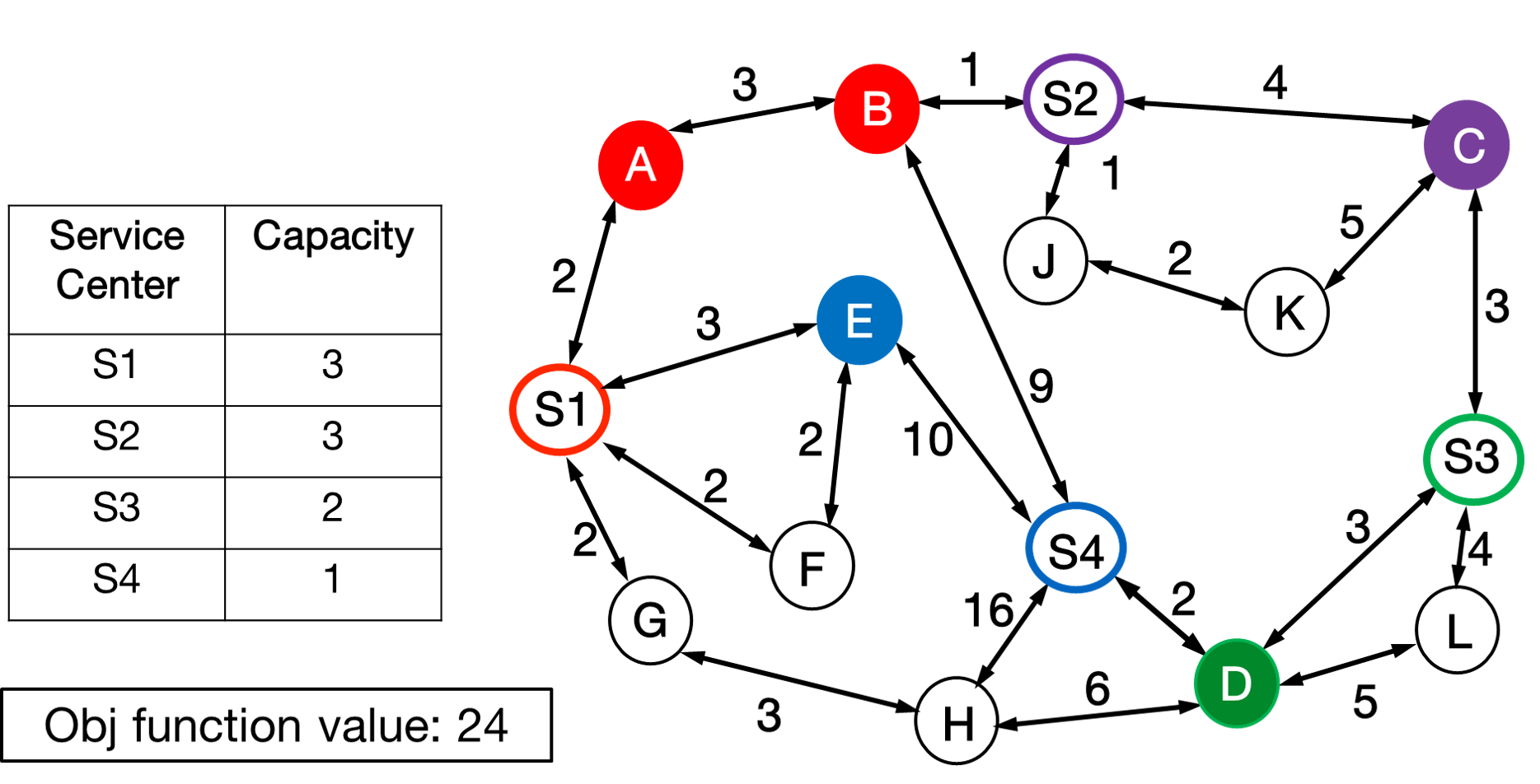}}
\subfigure[Sample edges in the allotment sub-space multi-graph]{\label{ncycb}\includegraphics[width=0.25\textwidth]{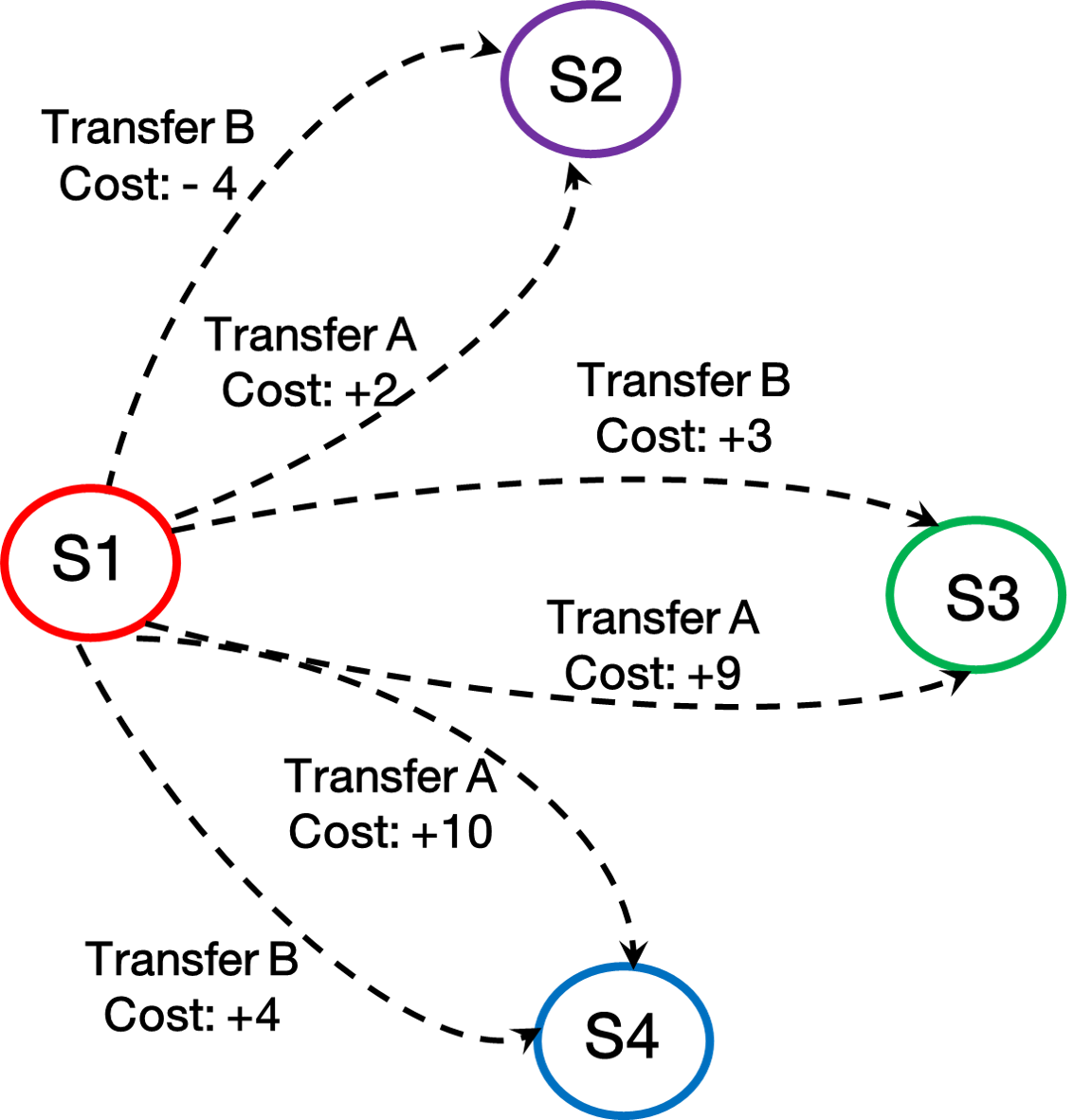}}
\subfigure[A negative cycle]{\label{ncycc}\includegraphics[width=0.30\textwidth]{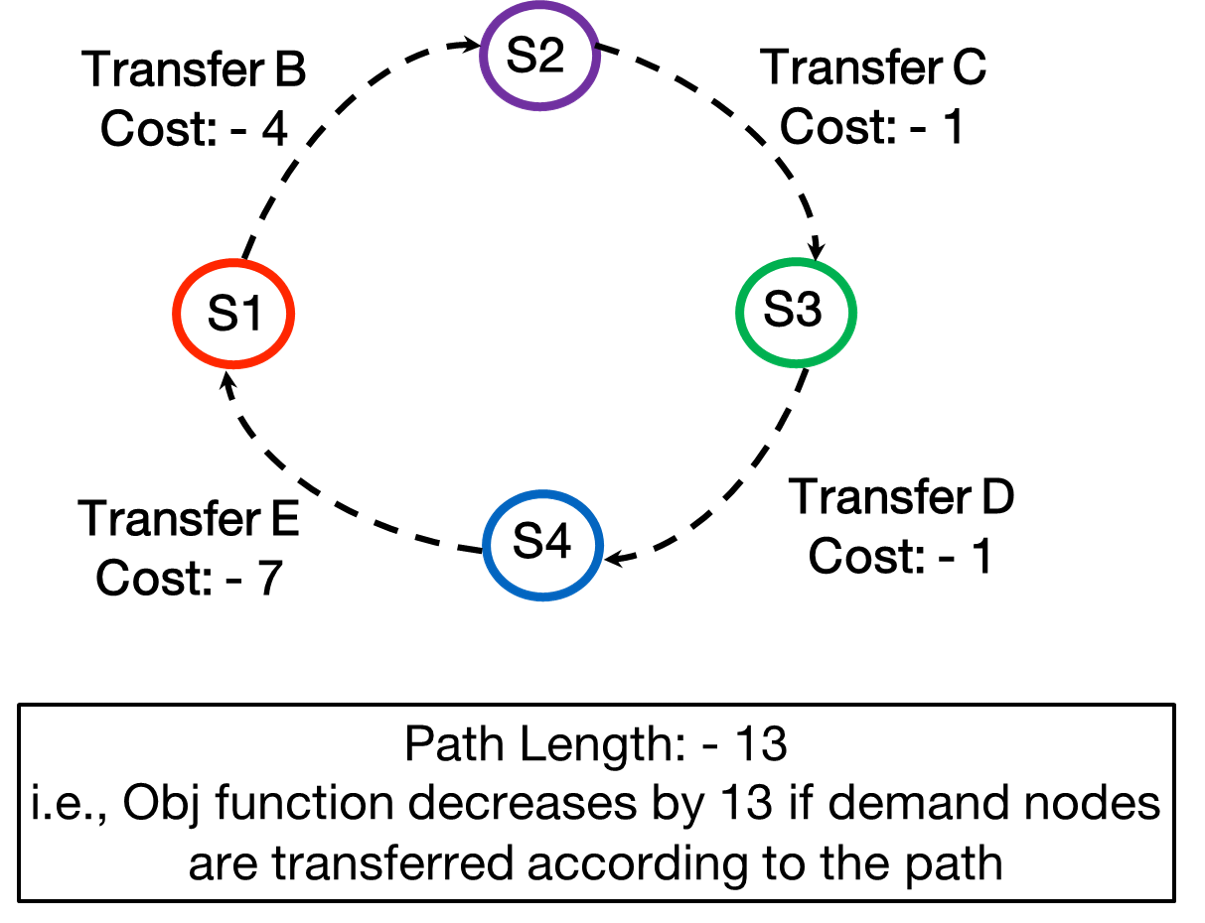}}
\end{center}
\caption{Illustrating allotment sub-space multi-graph on sample allotment. Vertices S1, S2, S3 and S4 are service centers. demand nodes allotted to a service center are shaded using the same color as that of the service center, e.g., demand node A is allotted to S1 in the sub-figure (a). }
\label{ncyc}  
\vspace{-2mm}   
\end{figure*}   

\noindent \textbf{Adjustment via Negative Cycle Removal:} Consider the cycle shown in Figure \ref{ncycc}. This directed cycle appears in the allotment subspace multigraph of the allotment shown in Figure \ref{ncyca}. Total cost of this cycle happens to be $-13$. In other words, we were to we adjust our allotment (shown in Figure \ref{ncyca}) by transferring the demand nodes corresponding to the edges in the cycle (e.g., transfer $B$ to S2, transfer $C$ to S3, transfer $D$ to S4, transfer $E$ to S1) then, the objective function value would decrease by $13$ units.  

\begin{figure}[h]  
\begin{center} 
\subfigure[Sample allotment.]{\label{npa}\includegraphics[width=80mm]{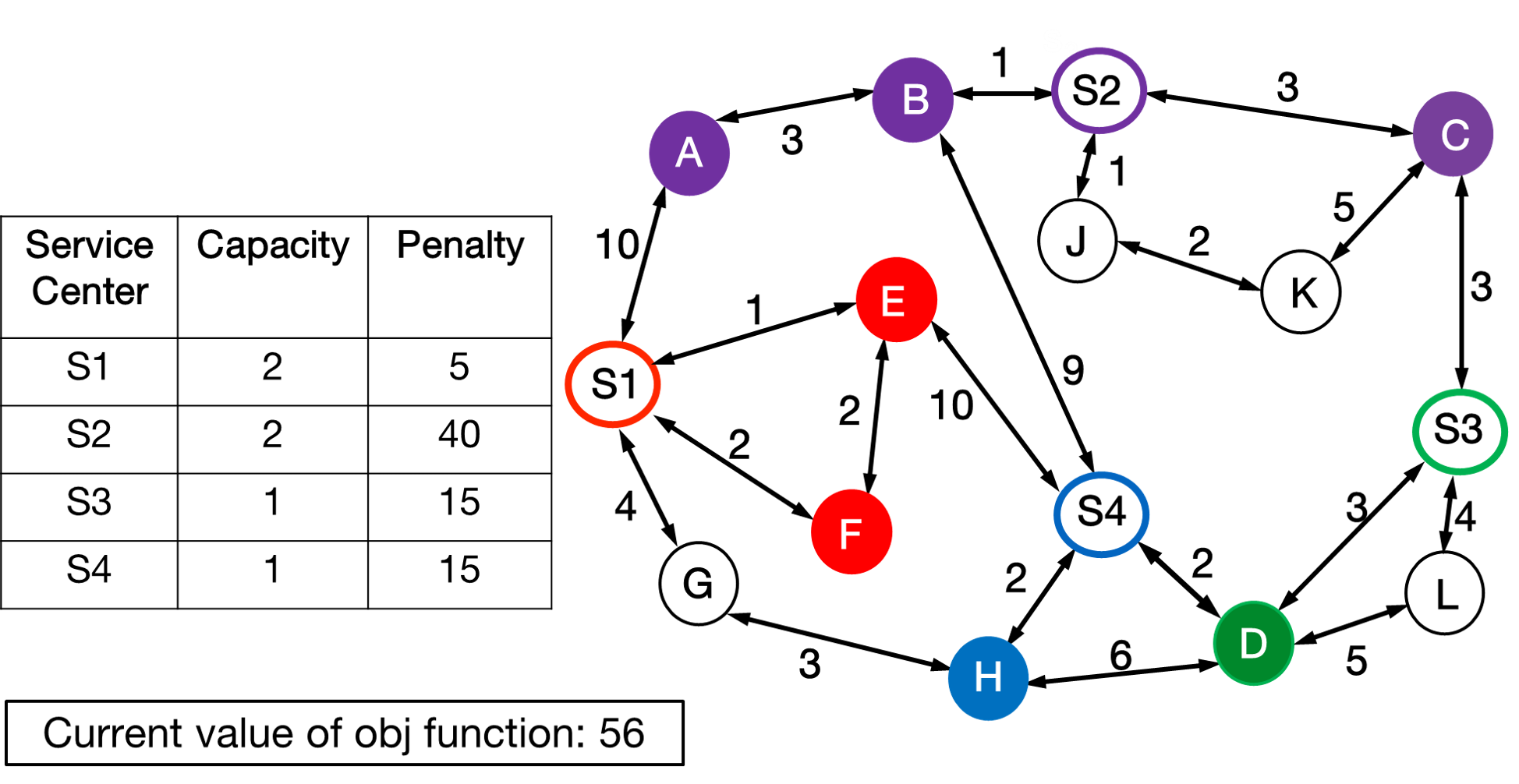}}
\subfigure[A negative path in the allotment sub-space graph]{\label{npb}\includegraphics[width=60mm]{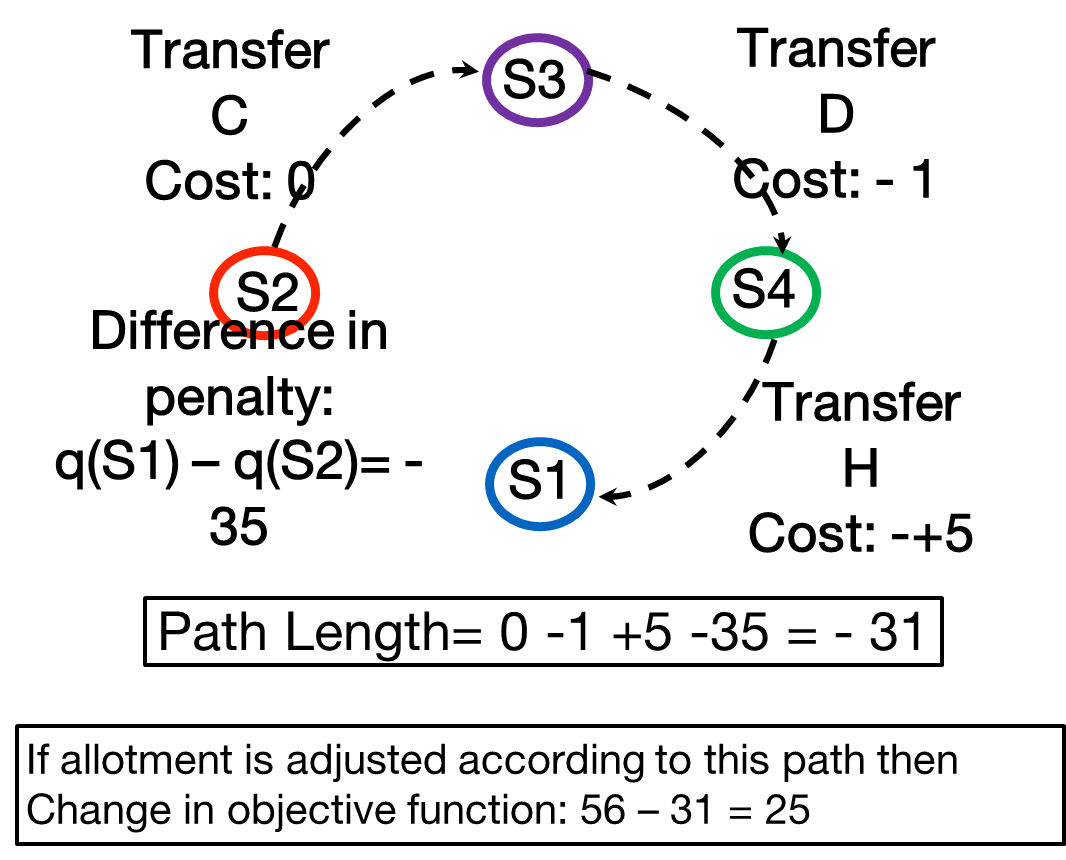}}
\end{center}
\caption{Illustrating negative paths in allotment sub-space multi-graph. Vertices S1, S2, S3 and S4 are service centers. Demand nodes allotted to a service center are shaded using the same color as that of the service center.}
\vspace{-2mm}   
\end{figure} 

\noindent \textbf{Adjustment via Negative Path Removal:}  Consider the allotment shown in Figure \ref{npa}. Figure \ref{npb} illustrates a path in its corresponding allotment subspace multigraph. On this path, we transfer node $C$ to S3, $D$ to S4 and finally, transfer $H$ to S1.  Also note that S2 was initially overloaded (in Figure \ref{npa}). Thus, if we are to adjust the allotment according the mentioned path, then the total penalty paid at S2 would decrease, since it is transferring a demand node ($C$) to another service center. On the other hand, we now have to pay the penalty at S1 as its accepting demand node beyond its capacity. Following is an expression of the total cost of this path:

\noindent Cost of Path = TransferCost(C~to~S3) + TransferCost(D~to~S4) + 

\hspace{1.5cm} TransferCost(H~to~S1) + Penaltyat(S1) - Penaltyat(S2)

\hspace{1.5cm} = \textbf{-31}      

Thus, if we adjust our initial allotment (shown in Figure \ref{npa}) according to the path shown in Figure \ref{npb}, then objective function value would decrease by 31. 

\noindent \textbf{ Loops:}
Cycles and paths can be generalized into the 
concept of {\bf  Loop.} A loop, induced by the allotment $ \mathcal{A}$,  starting from $s_i$ and ending in  $s_j$ ($s_i$ to $s_j$),  can be defined by any one of the following: \\
{\em   (i)}  a path from  $s_i$ to $s_j$ in $\Gamma(\mathcal{A})$, and an edge from $s_j$ to $s_i$ in $\Gamma(\mathcal{A})$. Such a loop is also a cycle in  $ \mathcal{A}$, containing an edge from  $s_i$ to $s_j$.  \\
{\em   (ii)} a path from  $s_i$ to $s_j$ in $\Gamma(\mathcal{A})$, and the edge from $s_j$ to $s_i$ in $\Gamma_{pen}(\mathcal{A})$.  Note that in this case, the penalty transfer edge is used to complete the loop.\\
If one of the two costs is negative, we say the the allotment $ \mathcal{A}$ induces a {\em negative cost loop} starting at $s_i$.

 A negative loop from $s_i$ to $s_j$ is removed just like we remove a negative cost path or   a negative cost cycle, depending on how we chose the edge from $s_j$ to $s_i$.

\begin{proposition}
Let $\mathcal{A}$ be any allotment. If  $ \mathcal{A}$ induces a non-simple negative loop, then $\mathcal{A}$   also induces a simple negative loop.   It therefore suffices to show the non-existence of \emph{simple} loops of negative cost. 
\label{simplecyc}
\end{proposition}

\begin{proposition}
 
A directed graph, $G(V,E)$, admits a decomposition
into directed cycles if and only if,    $\forall v \in V$, $indegree(v) = outdegree(v)$.  
\label{digraphdecomp}
\end{proposition}
\begin{theorem}
\label{optimalityThm}
  $\mathcal{A}$ is an optimal allotment if and only if  $\mathcal{A}$ does not induce any negative cost loops.
\end{theorem}

  \begin{proof}  
Since the removal of a negative cost loop reduces the total cost, it is obvious that an absence of negative cost loop is necessary for optimality. We shall now prove that an absence of negative cost loop is sufficient for optimality. 
 
Let $\mathcal{A}^{opt}$ be an optimal allotment, and  $\mathcal{A}$ be any allotment   that   does not induce  any negative cost cycles.  We shall show that if $Cost(\mathcal{A}^{opt})$ is less than $Cost(\mathcal{A})$, then $\mathcal{A}$ must induce a negative cost cycle. \\
We construct the {\em Difference Multigraph} for $\Gamma(\mathcal{A})$ and  $\Gamma(\mathcal{A}^{opt})$, denoted  $DM(\mathcal{A}^{opt},  \mathcal{A})$ as follows: \\
(1) For each service center $s_i, 1 \leq i \leq n_s+1$, add a vertex $v_i$ to $DM(\mathcal{A},  \mathcal{A}^{opt})$.\\
(2) Let $d_i, 1 \leq i \leq n_d,$ be any demand node such $\mathcal{A}$ assigns $d_i$ to $s_p$, $\mathcal{A}^{opt}$ assigns $d_i$ to $s_q$, and $p \neq q$. $DM(\mathcal{A},  \mathcal{A}^{opt})$ contains an edge, $(v_p, v_q, d_i)$
of weight $\mathcal{CM}(d_i, s_p) - \mathcal{CM}(d_i, s_q)$.

  $DM(\mathcal{A},  \mathcal{A}^{opt})$ represents the set of all demand node movements needed to transform allotment $\mathcal{A}$ into  $\mathcal{A}^{opt}$, i.e.,  all these edges must belong to $\Gamma(\mathcal{A})$.
  We break this up into two cases to simplify the presentation:\\
{\em Case (i) Every service center has the same occupancy in both  $\mathcal{A}$ and $\mathcal{A}^{opt}$. }

  All the service centers  have the same occupancies in both allotments, Thus, when we transform $\mathcal{A}$ into $\mathcal{A}^{opt})$ the number of demand nodes moved out of any particular service center is the same as the number of demand nodes moved into that service center. This means that  the out-degree of each node in   $DM(\mathcal{A},  \mathcal{A}^{opt})$ must be the same as its in-degree. By Proposition \ref{digraphdecomp}, this multigraph can   be partitioned into a set of edge-disjoint directed cycles.  Since $\mathcal{A}$ has greater cost than $\mathcal{A}^{opt}$, the sum of all the edge weights of $DM(\mathcal{A},  \mathcal{A}^{opt})$ must be less than zero, and therefore at least one of these cycles must have cost less than zero. All these edges must belong to   $\Gamma(\mathcal{A})$, i.e., $\mathcal{A}$ induces a negative cycle.
 
{\em Case (ii)  Every service center does not have the same occupancy in  $\mathcal{A}$ as in $\mathcal{A}^{opt}$. } In this case, it is possible that some nodes in $DM(\mathcal{A},  \mathcal{A}^{opt})$ have a greater in-degree, whereas others have a greater out-degree. 
Consider the following process: {\em Starting at any demand node $s_x$ which had a greater out-degree, perform a depth-first search. If the search revisits a vertex, we have a directed cycle which is removed. Otherwise, the search ends in a vertex $s_y$ which has no outgoing edges. All the edges in the depth-first search stack form a  path from $s_x$ to $s_y$; we can therefore add the penalty transfer edge from $s_y$ to $s_x$ to get a loop.}
The above process decomposes $DM(\mathcal{A},  \mathcal{A}^{opt})$ into a collection of loops. 
The removal of all these loops transforms $\mathcal{A}$ into  $\mathcal{A}^{opt}$, and therefore the total cost of all the loops is equal to the difference   $Cost(\mathcal{A}^{opt})- Cost(\mathcal{A})$. 
Since this quantity is less than zero, at least one of the loops must have negative cost.

\end{proof}

 \subsubsection{Procedure for Removing Negative Loops}
 \label{ncalgo}
 
  To compute the most  negative loop that {\em starts at} $s_j$  using Bellman-Ford, create a variant of the allotment subspace multigraph, called  the \emph{NegLoop} allotment subspace graph, as explained  below.

\begin{definition}
Given an allotment  $\mathcal{A}$, and a distinguished service center $s_j$, the Negloop allotment subspace graph, $NLoop(\mathcal{A}, s_j) = (P,Q)$, is defined as follows:\\ 
\textbf{ (1) Set of nodes} ($P$): Each service center $s_l \in S-{s_j}$ is a node in $P$. Service center $s_j$ (referred to as the {\em anchor} or {\em distinguished} service center) is represented by two nodes $s_j^{in}$ and $s_j^{out}$. \\
\textbf{ (2) Set of edges} ($Q$): Cost of the edges between any two nodes $u$ , $v$ $\in P$ is defined as follows:  \\
        (a) if $u \notin \{s_j^{in},s_j^{out}\}$ and $v \notin \{s_j^{in},s_j^{out}\}$, then $cost(u,v)$ is   the cost of $\Gamma_{min}(u, v)$. \\
        (b) if $u = s_j^{out}$, then $cost(u,v)$ is the cost of $\Gamma_{min}(s_j, v)$ . \\
        (c) if $v = s_j^{in}$, then $cost(u,v)$ is  the lesser of the costs of $\Gamma_{min}(u, v)$ and $\Gamma_{pen}(u, v)$.\\
        (d) There are no other edges in $Q$.
\label{minLoopGamma} 
\end{definition}

After the creating of the NegLoop allotment sub-space graph, we determine the lowest cost path between the "in" and "out" copies ($s_j^{in}$ and $s_j^{out}$) of the anchor service center. Note that since both $s_j^{in}$ and $s_j^{out}$ map to $s_j$ in $\Gamma(\mathcal{A})$, this path gives us the loop of least cost starting at $s_j$.  For computing the lowest cost path from $s_j^{out}$ to $s_j^{in}$, we use Bellman Ford's label correcting algorithm.  Note that we cannot use a label setting approach (e.g., Dijkstra's) for finding lowest cost path   as the graph may contain negative edges.  

\begin{proposition}
The shortest path from   $s_j^{out}$ to $s_j^{in}$ in $NLoop(\mathcal{A}, s_j)$ can be computed using Bellman-Ford, if $\mathcal{A}$ does not induce any negative cost loops that do not pass through $s_j$. Combining $s_j^{out}$ and $s_j^{in}$ into one vertex, this shortest path gives us the most negative loop starting at $s_j$ in the allotment $\mathcal{A}$. The entire process can be completed in $O(k^3 + k^2 \log n)$ steps.
\end{proposition}

We can modify our process  to obtain the loop of least cost that  {\em passes through} $s_j$, as follows:\\
\begin{enumerate}
    \item { \em Change the construction of the Negloop allotment graph.} The cost of each edge $(u, v)$ in $NLoop(\mathcal{A}, s_j)$ is defined as the lesser of the weights of $\Gamma_{min}(u, v)$ and $\Gamma_{pen}(u, v)$.
    \item { \em Extract the negative loop from the shortest path.} Once we find the shortest path in  $NLoop(\mathcal{A}, s_j)$ as defined above, we get a shortest path from $s_j^{out}$ to $s_j^{in}$ that may contain multiple edges from $\Gamma_{pen}()$. Let the first such edge on the shortest path  be from $s_x$ to $s_y$, and the last such edge on the shortest path  be from $s_z$ to $s_w$. Replace the sequence of edges from $s_x$ to $s_w$ by one edge $\Gamma_{pen}(x, w)$; the resulting loop, after merging $s_j^{in}$ and $s_j^{out}$ into one vertex, is the loop of least cost passing though $s_j$. 
\end{enumerate}
The above process can clearly be completed in $O(k^3 + k^2 \log n)$ steps. For correctness, 
note that replacing $\Gamma_{pen}(x, y)$ and $\Gamma_{pen}(z, w)$ by $\Gamma_{pen}(x, w)$ and $\Gamma_{pen}(z, y)$ does not increase the total cost. Since $\mathcal{A}$ did not have any negative loops  not passing through $s_j$, we have \\
Cost($\Gamma_{pen}(z, y)$) + Cost(edges from $s_y$ to $s_z$)  $ \geq 0$.\\
In other words,\\
Cost(edges from $s_j^{out}$ to $s_x$) + Cost($\Gamma_{pen}(x, w)$) + Cost(edges from $s_w$ to $s_j^{in}$) \\
must equal the cost of the shortest path found by Bellman-Ford. Therefore, we have:
\begin{proposition}
The   most negative passing through $s_j$ in the allotment $\mathcal{A}$ can be computed in $O(k^3 + k^2 \log n)$ steps.
\end{proposition}

 \section{Maintaining an Optimal Allotment}
 \label{optsec}
 We maintain an optimal allotment by removing all negative loops every time there is a change to the problem instance. Following are some examples of the changes we can make:
 \begin{enumerate}
     \item {\em Adding a new demand node to some service center.}
     \item {\em Removing an existing demand node (also unassigning it from a service center) }
     \item {\em Making changes to the capacity or penalty function for a service center }
     \item {\em Making changes to the cost matrix}
 \end{enumerate}
 
 In each of these cases, the nature of the allotment changes when the problem instance changes. Let $\mathcal{P}1$ be the problem instance prior to the changes, and let   $\mathcal{A}1$ be an  optimal allotment for $\mathcal{P}1$.   Let $\mathcal{P}2$ be the problem instance after making the changes. 
 Changing a penalty function or the cost matrix does not change the allotment, but 
 adding or removing demand nodes results in an new allotment. Without loss of generality, say  that we have a problem instance $\mathcal{P}2$ with an allotment $\mathcal{A}2$. As a result, there are  changes to $\Gamma_{min}() $  and/or $\Gamma_{pen}()$.   These changes can result in the formation of negative loops; the allotment needs to be modified so that the resulting allotment $\mathcal{A}3$ is free of negative loops.

  We shall show now examine all these modifications and show that for some restricted kinds of changes, we can compute a new optimal allotment in $O(k^3 + k^2\log n)$ time.  In all these cases, it can be shown that removing the most negative loop induced by $\mathcal{A}2$, yields an allotment $\mathcal{A}3$ that is optimal for problem instance $\mathcal{P}2$.
  The requirements we place on the modification to problem instance $\mathcal{P}1$ are the following:
  
  \begin{enumerate}
    
      \item {\em All the modifications are centered around one service center $s_j$. }
      We refer to this as the distinguished service center. If we add or remove  a demand node or modify  a penalty function this is clearly met. If the cost matrix is modified, then these changes should all be connected to one service center. 
      \item {\em All affected edges lead into or lead out of the distinguished service center.} When we increment (decrement) capacity of a service center, only the outgoing (incoming, resp) penalty edges
      of that service center are affected. When we insert a demand node into $s_j$, this condition gets violated, since outgoing transfer edges and incoming penalty edges of $s_j$ are affected. We will discuss this in the next subsection.  
      \item {\em Changes are confined to one demand node. } If the change is such that more than one demand node is independently affected, then we do not have a way to recompute the optimal solution as claimed. We will elaborate on this later.

  \end{enumerate}

  \subsection{Adding a new demand node}
Say that we have an allotment $\mathcal{A}1$, and we add a  demand node,  $d_i$, to service center $s_j$, the structure of the resulting allotment, $\mathcal{A}2$,  is different in two ways:\\
{\em (i)  The weight of any edge leading out of $s_j$ in $\Gamma_{min}() $ may decrease. } This happens because we have a new demand node in $s_j$, which may be closer to some service center $s_l$, than any of the demand nodes currently in $s_j$.\\
{\em (ii) The weight of any edge leading into $s_j$ in $\Gamma_{pen}() $ may decrease. } This happens because we have a new demand node in $s_j$ which increases its occupancy and hence increases the penalty associated with $s_j$. 

If we have to maintain an allotment that is free of negative loops, we need  a procedure to remove the negative loops from $\mathcal{A}2$. As mentioned earlier,  this modification to the problem instance $\mathcal{P}1$  does not comply with the restrictions we placed. 
However, we can treat the addition of a demand node as two independent sequentially occurring changes: {\em first, change the weights of the outgoing edges from $s_j$ in $\Gamma_{min}(\mathcal{A}1)$, and remove any resulting negative loops; next increase the occupancy of $s_j$, update the weights in $\Gamma_{pen}()$, and remove any resulting negative loops.}

Splitting the insert operation in this manner requires some justification, since changing the weights of the outgoing edges from $s_j$ may not represent an allotment for any problem instance. To see why this is admissible,
consider the following modified definition for the problem:\\
{\em (1) Add a dedicated dummy demand node associated with each service center.} For each service center $s_i$, we add a dummy demand node $d_{n_d+i}$, such $\mathcal{CM}(n_d+i, i)$ equals zero, and  $\mathcal{CM}(n_d+i, j)$ equals $\infty$ for all $j \neq i$.\\
{\em (2) Modify the penalty function to accommodate the dedicated demand node at zero cost.} This is accomplished by increasing the capacity of all service centers by one, and shifting the penalty function one unit to the right of the $X$-axis.

With these modifications, we can now treat the  addition of a demand node as two-step process:\\
{\em Step 1. Add the demand node $d_i$ to service center $s_j$, and remove the demand node $d_{n_d+j}$  from  $s_j$.  } In this step, we potentially reduce the costs of some edges leading out of $s_j$ in $\Gamma_{min}() $, but do not change the occupancy of $s_j$\\
{\em Step 2. Add the demand node   $d_{n_d+j}$  to  $s_j$.  }  In this step, we increase the occupancy of$s_j$, and  potentially reduce the costs of some edges leading in to  $s_j$ in $\Gamma_{pen}() $\\

Note that if we remove the negative loops after Step 1, we get an optimal allotment for a valid problem instance, and hence the splitting of the insert operation is justified.  

\subsection{Restoring Optimality after Insertion}

We shall now prove two technical lemmas that establish the correctness of our insertion process. In Step 1 of the insertion process, we are potentially reducing the costs of some edges leading our of $s_j$ in $\Gamma_{min}() $. In Lemma \ref{mainlemma} we prove that by removing the most negative loop passing through $s_j$ we get an optimal allotment. In Step 2,  we potentially reduce the costs of some edges leading in to  $s_j$ in $\Gamma_{pen}() $. In Lemma \ref{mainlemma2} we prove that by removing the most negative loop starting in $s_j$ we get an optimal allotment.

 \begin{lemma}
Let $ \mathcal{A}1$ be an allotment   that has no negative cost loops. 
Say that we modify $\mathcal{A}1$ as follows: \\
(1)  identify a distinguished vertex $s_j$ and a demand node $d$ that is assigned to $s_j$ in $\mathcal{A}1$. \\
(2) Replace $d$ with a demand node $d_i$, with arbitrary values for  $\mathcal{CM}(d_i,s_x)$, for all $1 \leq x \leq n_s$. \\ 
Let $\mathcal{A}2$ be the resulting allotment, and
let $C_{min}$ be the loop of least cost induced by $\mathcal{A}2$ that passes through $s_j$. We have the following:\\
(a) If $C_{min}$ has cost greater than or equal to zero, then $ \mathcal{A}2$ does not induce any negative loops.\\
(b) Say that $C_{min}$ has cost less than zero, and that modifying $ \mathcal{A}2$ to remove $C_{min}$   results in the allotment  $ \mathcal{A}3$. $ \mathcal{A}3$  does not induce any  negative loop.
\label{mainlemma}
\end{lemma}

Note that the only change caused by this modification, is that weights of edges leading out of $s_j$
in $\Gamma_{min}(\mathcal{A}2)$ may be lower than weights of corresponding edges in $\Gamma_{min}(\mathcal{A}1)$.
Before formalizing the proof details, it is useful to understand the  construction we employ. 
 The proof is by contradiction, i.e., we show that if  $ \mathcal{A}3$ has a negative loop (say $C_{neg}$) then there must have been a   a negative loop in  $ \mathcal{A}1$. We use the edges from $C_{min}$ and $C_{neg}$ to find a  loop, $C_{cont}$, in $\mathcal{A}1$ with total cost less than zero.  The challenge here is that $C_{min}$  consists of edges from  $\Gamma(\mathcal{A}2)$ and $\Gamma_{pen}(\mathcal{A}2)$, and $C_{neg}$ consists of edges from $\Gamma(\mathcal{A}3)$ and $\Gamma_{pen}(\mathcal{A}3)$. Since $C_{cont}$ can have edges only from $\Gamma(\mathcal{A}1)$ and $\Gamma_{pen}(\mathcal{A}1)$,    we need some additional processing.
 
{\bf  Nullifying edge dependencies.} \\
(a)  Two transfer edges, $e_1$   and $e_2$, are {\bf dependent} if they have the form $e_1 = (s_p, s_q, d_x)$ and $e_2 = (s_q, s_r, d_x)$, i.e., $e_1$  transfers  $d_x$ from $s_p$ to $s_q$ and $e_2$ transfers the same demand unit $d_x$ from $s_q$ to $s_r$. Note that a dependency can happen only when $e_1$   and $e_2$ belong to different allotment subspace multigraphs. (In our case $e_1$ is from $C_{min}$, i.e., $\mathcal{A}2$, and  $e_2$ is from $C_{neg}$, i.e., $\mathcal{A}3$.) This dependency is {\bf nullified} by removing $e_1$ and $e_2$, and adding the edge $e_3 = (s_p, s_r, d_x)$. Note that $e_3$ is from $\mathcal{A}2$, and the cost of $e_3$ is the sum of the costs of $e_1$ and $e_2$, i..e, {\em the total cost   remains unchanged when dependencies are nullified}.\\
 (b)   For dependencies between penalty edges, here are two situations to consider. The first situation occurs when  the loop $C_{min}$ has a penalty edge   $e_1(s_x, s_y)$, and $C_{neg}$ has a penalty edge   $e_2(s_y, s_z)$. 
     $ weight(e_{1}) = q_{x}(o_x(\mathcal{A}2) + 1) - q_{y}(o_y(\mathcal{A}2))$ \\
  $ weight(e_{2}) = q_{y}(o_y(\mathcal{A}3) + 1) - q_{z}(o_z(\mathcal{A}3))$ \\
  Going from $\mathcal{A}2$ to $\mathcal{A}3$ the occupancy of $s_y$ decreases by one, and the occupancy of $s_z$ remains unchanged. Therefore,  \\
   $ weight(e_{2}) = q_{y}(o_y(\mathcal{A}2)) - q_{z}(o_z(\mathcal{A}2))$ \\
   The sum of the weights of the edges is thus \\
   $ weight(e_{1}) + weight(e_{2}) = q_{x}(o_j(\mathcal{A}2) + 1) - q_{z}(o_z(\mathcal{A}2))$ \\
   which equals the weight of penalty edge  from $s_x$ to $s_z$. Thus we can replace $e_1 $ and  $e_2 $
 in  $C_{min} \bigcup C_{neg}$ by a single penalty edge  from $s_x$ to $s_z$. Alternately, if  $C_{min}$ has a penalty edge   $e_2(s_y, s_z)$, and $C_{neg}$ has a penalty edge   $e_1(s_x, s_y)$\\
  $ weight(e_{2}) = q_{y}(o_j(\mathcal{A}2) + 1) - q_{z}(o_y(\mathcal{A}2))$ \\ and \\
 $ weight(e_{1}) = q_{x}(o_x(\mathcal{A}3) + 1) - q_{y}(o_y(\mathcal{A}3))$ \\  
$ = q_{x}(o_x(\mathcal{A}2) + 1) - q_{y}(o_y(\mathcal{A}2 + 1))$. \\
 Once again, 
 $ weight(e_{1}) + weight(e_{2}) = q_{x}(o_x(\mathcal{A}2) + 1) - q_{z}(o_y(\mathcal{A}2)) $
 
  A second kind of dependency occurs when we have \\
  {\em (i)} $C_{min}$ with   $e_1(s_x, s_y)$, and $C_{neg}$ with   $e_2(s_x, s_z)$, or \\
  {\em (i)} $C_{min}$ with   $e_1(s_x, s_y)$, and $C_{neg}$ with   $e_2(s_z, s_y)$ \\
  In such a situation, the weight of edges can depend on the order in which we carry out the occupancy transfer. This is not nullified, but we will show that this order can be chosen appropriately in each case to achieve the needed end result.

 In our quest for $C_{cont}$,  we start with the edges in $C_{min}$ and $C_{neg}$ and nullify  the dependencies as explained above.
 Figure \ref{lemfig} shows an example of how this proof works. The edges $(S1, Sa)$ and $(S4, Sd)$ in $C_{neg}$ are dependent, respectively, on edges $(S^*, S1)$ and $(S3, S4)$ in $C_{min}$. In nullifying these dependencies, we  add the edges $(S^*, Sa)$ and $(S3, Sd)$, which  leaves us with two independent cycles. The distinguished edge $(S^*, S1)$ is replaced by the distinguished edge  $(S^*, Sa)$, and all the edges in the inner cycle $S1 \rightarrow S2 \rightarrow S3 \rightarrow Sd \rightarrow S1$ belong to $\Gamma(\mathcal{A}1)$. The outer cycle $S^* \rightarrow Sa \rightarrow Sb \rightarrow Sc \rightarrow S4 \rightarrow S5 \rightarrow S^*$ contains the   only distinguished edge, and all  its edges are present in $\Gamma(\mathcal{A}2)$. Therefore,  it must have cost greater than or equal to that of $C_{min}$.  Dependency nullification preserves the cost, and therefore the total cost of the inner cycle $S1 \rightarrow S2 \rightarrow S3 \rightarrow Sd \rightarrow S1$ has to be less than or equal to that of $C_{neg}$, i.e., it has to be negative. Thus the inner cycle is the $C_{cont}$ we are looking for.
\begin{figure}[ht] 
\includegraphics[width=80mm]{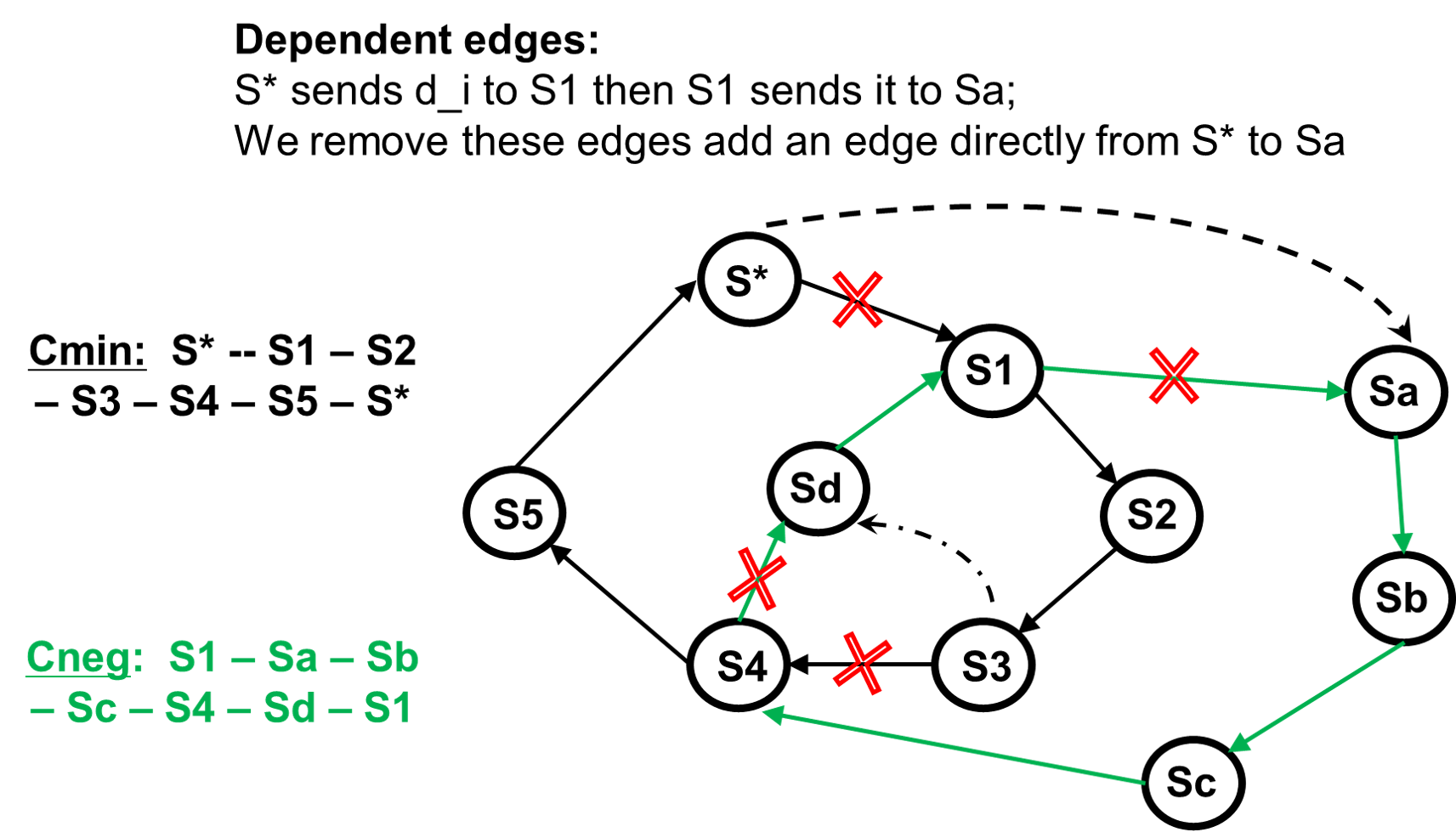}
\caption{Example illustrating   proof of Lemma \ref{mainlemma}.}
\label{lemfig}  
\vspace{-2mm}   
\end{figure}

\begin{proof} (of Lemma \ref{mainlemma}) 
 By Proposition \ref{simplecyc} the distinguished vertex appears only once in $C_{min}$. Hence $C_{min}$ contains exactly one   edge leading out of the distinguished vertex, and  this edge moves $d_i$ from $s_j$ to some service center $s_p$   All pairs of edges in $C_{min}$ are independent, as are all pairs of edges in $C_{neg}$. Consider the digraph $H(V,E)$, defined by $C_{min} \bigcup C_{neg}$. Since $H$ is a union of two cycles, it satisfies the property that $\forall v \in V, indegree(v) = outdegree(v)$. $H$ will have pairs  of edges, $(e_1 \in C_{min}, e_2 \in C_{neg})$, that are dependent. Let $H'(V,E')$ be the resulting digraph, after we nullify the dependencies in $H$. For any dependency, the vertex, $v_q$, corresponding to $s_q$, must belong to both cycles, i.e., $indegree(v_q) = outdegree (v_q) = 2$. After nullifying the dependency, we have $indegree(v_q)$ $= outdegree (v_q) = 1$.  Hence $H'$ also satisfies the property that $\forall v \in V, indegree(v) = outdegree(v)$, i.e. $H'$ admits a decomposition into directed cycles. The following claims can be established: 

{\bf Claim 1:} {\em All the demand node transfer edges in $H'$ belong to $\Gamma(\mathcal{A}2)$.} The only way this property can be violated, is when $C_{neg}$ contains an edge that depends on an edge in $C_{min}$. Since all dependencies are nullified, the claim holds for these edges.  

{\bf Claim 2:} {\em All the demand node transfer edges in $H'$ belong to $\Gamma(\mathcal{A}1)$, except for the modified edge.} The edges in $C_{min}$ are from  $\Gamma(\mathcal{A}1)$, except for the modified edge. The edges in $C_{neg}$ that are not in $\Gamma(\mathcal{A}1)$, are all dependent on edges in $C_{min}$, and are replaced by edges from $\Gamma(\mathcal{A}1)$ when we nullify dependencies.   Since no occupancy is modified when we go from   $ \mathcal{A}1$ to $ \mathcal{A}2$,
the penalty transfer edges can also be assumed to be from $ \mathcal{A}1$.

{\bf Claim 3:} {\em Total cost of all edges in $H'$ is less than the total cost of the edges in $C_{min}$.} The edges in $C_{neg}$ have total cost less than zero, and dependency nullification preserves the cost.

To complete the argument, find a loop, $C1$, in $H'$,   that passes through   $s_j$, and includes the   edge that moves $d_i$. This loop must must have cost greater than or equal to cost of $C_{min}$, since $C_{min}$ was the loop of least cost in  $\mathcal{A}2$. Therefore, $H' - C1$ has cost less than zero, i.e., decomposing it into directed cycles yields at least one loop, $C_{cont}$,  of negative cost. To complete the argument, we need to show that we had a negative loop in $\mathcal{A}1$. We have the following cases:

{\em case (i) All the  edges in $C_{cont}$ are demand node transfer edges.} Since the modified edge is in  $C1$, 
by Claim 2, all  the edges in $C_{cont}$ are from $\mathcal{A}1$, i.e., this negative loop existed in $\mathcal{A}1$.\\
{\em case (ii)  $C_{cont}$ contains exactly one penalty transfer edge.} Let $e_1(s_x, s_y)$ be the penalty transfer edge. Since $e_1$ is not dependent on any other edge in $H'$, the occupancies of $s_x$ and $s_y$   in $\mathcal{A}3$
as they were in $\mathcal{A}1$,   i.e, the negative loop $C_{cont}$ existed in $\mathcal{A}1$.\\
{\em case (iii)  $C_{cont}$ contains two independent penalty transfer edges.} This can happen if both $C_{min}$ and 
$C_{neg}$ have a penalty transfer edge. Let $e_1(s_x, s_y)$ and $e_2(s_z, s_w)$ be the two edges, such that  $C_{cont}$ passes through $s_x$, $s_y$, $s_z$  and $s_w$ in that order. Since $ s_x, s_y, s_z$  and $s_w$ have the same occupancy as in $\mathcal{A}1$, we have: \\ 
  $ weight(e_{1}) = q_{x}(o_x(\mathcal{A}1) + 1) - q_{y}(o_y(\mathcal{A}1))$ \\
  $ weight(e_{2}) = q_{z}(o_z(\mathcal{A}1) + 1) - q_{w}(o_w(\mathcal{A}1))$ \\
 If we replace $e_1$ and $e_2$ by the edges $e_1(s_x, s_w)$ and $e_2(s_z, s_y)$ with weights: 
 $ weight(e_{3}) = q_{x}(o_x(\mathcal{A}1) + 1) - q_{w}(o_w(\mathcal{A}1))$ \\
  $ weight(e_{2}) = q_{z}(o_z(\mathcal{A}1) + 1) - q_{y}(o_y(\mathcal{A}3))$ \\
  This gives us two loops, with same total weight as $C_{cont}$, which are from  $\mathcal{A}1$. Since $C_{cont}$
  has total weight less than zero, at least one of the two loops must have weight less than zero. \\
 {\em case (iv)  $C_{cont}$ contains a  dependent penalty transfer edge.} This can happen in one of two ways:\\
 (a) $C_{cont}$ has an edge $e_1(s_x, s_y)$ and  there is another edge $e_2(s_x, s_z)$ in $H'$. If $e_1$ is from 
 $\mathcal{A}2$, it has the same weight as it had in $\mathcal{A}1$, and our argument is complete. \\
 (b) If  $e_2$ is from 
 $\mathcal{A}2$, and $e_1$ is from  $\mathcal{A}3$, we have: \\
 $ weight(e_{2}) = q_{x}(o_x(\mathcal{A}1) + 1) - q_{z}(o_z(\mathcal{A}1))$. \\
 Since the occupancy of $s_x$ increases by one from $\mathcal{A}2$ to $\mathcal{A}3$  \\
 $ weight(e_{1}) = q_{x}(o_x(\mathcal{A}1) + 2) - q_{y}(o_y(\mathcal{A}1))$ \\
 However, since the penalty function is monotone increasing,  \\
 $ weight(e_{1}) in \mathcal{A}1  = q_{x}(o_x(\mathcal{A}1) + 1) - q_{y}(o_y(\mathcal{A}1)) \\ \leq  q_{x}(o_x(\mathcal{A}1) + 2) - q_{y}(o_y(\mathcal{A}1))$ \\
 Thus we have a negative loop in $\mathcal{A}1$.
\end{proof}

\begin{lemma}
Let $ \mathcal{A}1$ be an allotment   that has no negative cost loops. 
Say that we modify $\mathcal{A}1$  by
 identifying a distinguished vertex $s_j$ and modifying weights of edges in $\Gamma_{pen}(\mathcal{A}1)$,  that are directed to $s_j$ from all other nodes.
Let $ \mathcal{A}2$ be the resulting allotment, and let $C_{min}$ be the loop of least cost induced by $\mathcal{A}2$ that begins at $s_j$ and contains an   edge of modified weight. We have the following:\\
(a) If $C_{min}$ has cost greater than or equal to zero, then $ \mathcal{A}2$ induces no negative loops.\\
(b) Say that $C_{min}$ has cost less than zero, and that modifying $ \mathcal{A}2$ to remove $C_{min}$   results in the allotment  $ \mathcal{A}3$. $ \mathcal{A}3$  does not induce any  negative loop.
\label{mainlemma2}
\end{lemma} 
\begin{proof}
This situation can happen when the occupancy of $s_j$ is increased.
Following our proof of Lemma \ref{mainlemma}, $C_{min}$ will end with a penalty transfer edge, $e_1$, from $s_p$ (say) to $s_j$.  If there is no penalty transfer edge in $C_{neg}$, we identify a loop, $C1$,  in $H'$ that contains $e_1$;  among the remaining edges, by way of contradiction, we find the negative cycle $C_{cont}$.   

If $C_{neg}$ has a penalty transfer edge $e_2$, we have the following cases: \\
\emph{case (i)  $e_2$ is not incident on either $s_p$ or $s_j$.} $e_2$ will have the same weight as it did in $\mathcal{A}1$, and hence any  $C_{cont}$ that  contains $e_2$ existed in $\mathcal{A}1$. \\
\emph{case (ii)  $e_2$ is directed $s_q$ to $s_p$.} The dependency is nullified by replacing $e_1$ and $e_2$ with a single edge, $e_3$, from $s_q$ to $s_j$. $C1$ is chosen as a loop that passes through  $e_3$. Note that $e_3$ was a candidate edge in the computing of $C_{min}$, and hence $C1$ has cost greater than or equal to that of $C_{min}$. Thus there should be a negative cycle in $H' - C1$ \\ 
\emph{case (iii)  $e_2$ is directed $s_j$ to $s_q$.} The dependency is nullified by replacing $e_1$ and $e_2$ with a single edge, $e_3$, from $s_p$ to $s_q$.
As a result, $H'$ does not contain any edge with modified weight. Any negative cycle in $H'$ serves as $C_{cont}$. \\
\emph{case (iv)  $e_2$ is directed $s_q$ to $s_j$.} The removal of $C_{min}$ decreases the occupancy of $s_j$ by one, and therefore $q_{j}(o_j(\mathcal{A}3))$ is the same as $q_{j}(o_j(\mathcal{A}1))$, i.e., any $C_{cont}$ containing $e_2$ was a negative loop under $\mathcal{A}1$. \\
\emph{case (v)  $e_2$ is directed $s_p$ to $s_q$.} The removal of $C_{min}$ increases the occupancy of $s_p$ by one, and therefore $q_{p}(o_p(\mathcal{A}3))$ is greater than $q_{p}(o_p(\mathcal{A}3))$, i.e., the the weight of $e_2$ in $\mathcal{A}1$ is less than the weight of $e_2$ in $\mathcal{A}3$. Hence any $C_{cont}$ containing $e_2$ was a negative loop under $\mathcal{A}1$.
 
\end{proof}

We are now ready to prove the following theorem:
\begin{theorem}
Consider the process of maintaining the allotment of $n$ demand nodes to $k$ service centers such that there are no negative cost loops. The addition of a demand node can be completed in time $O(k^3 + k^2\log n)$.
\end{theorem}
 
 \begin{proof}
 Let $\mathcal{A}1$ be the allotment prior to adding the new demand node. The addition is performed as a two-step process. For Step 1, we start by updating $k-1$ $BestTransHeaps$, one for each outgoing edge from $s_j$ in $\Gamma_{min}(\mathcal{A}1)$. Let $\mathcal{A}2$ be the resulting allotment.    Using this, we construct $NLoop(\mathcal{A}2, s_j)$, the entire process taking  $O( k\log n + k^2)$ operations. We then find the shortest path from  $s_j^{out}$ to $s_j^{in}$ using Bellman-Ford, which takes 
 $O(k^3)$ steps. If this path has cost greater than or equal to zero, we move to Step 2. Otherwise, this path is either of the following:\\ 
 {\em case (i). } a negative cycle in $\Gamma_{min}(\mathcal{A}2)$ \\
 {\em case (ii). } a path from $s_j$ to some $s_x$, along with the penalty transfer edge  from $s_x$ to $s_j$. \\
 Removing this loop requires updating the   $BestTransHeaps$ of all service centers in the loop, i.e., update at most $k^2$ heaps, which is done in $O(k^2\log n)$ time.
 Let $\mathcal{A}3$ be the resulting allotment after removing the loop. We then move to Step 2 and perform a similar process.
 Thus the entire operation is accomplished in $O(k^3 + k^2\log n)$ steps.
  \end{proof}
  
  Since an instance of LBDD can be solved by setting up the data structures and inserting $n$ demand nodes, we have the following theorem:
  \begin{theorem}
An optimal allotment for an LBDD instance with $n$ demand nodes and $k$ service centers can be   computed in time $O(nk^3 + nk^2\log n)$.
\end{theorem}

 \subsection{Removing a demand node}
 Say that we have an allotment $\mathcal{A}1$, and we remove a  demand node,  $d_i$, from service center $s_j$, the   resulting allotment, $\mathcal{A}2$,  modifies $\mathcal{A}1$ in two ways:\\
{\em (i)  The weight of any edge leading out of $s_j$ in $\Gamma_{min}() $ may increase. } This happens because $d_i$ was the best  node in $s_j$ for moving to some service center $s_l$. An increase in edge weights cannot induce any negative loops that were not already induced by  $\mathcal{A}1$, we can ignore these modifications.\\
{\em (ii) The weight of any edge leading out of $s_j$ in $\Gamma_{pen}() $ may decrease. } This happens because we have lower occupancy in $s_j$ which decreases the penalty associated with $s_j$. 
 
 The next lemma is anti-symmetrical to Lemma \ref{mainlemma2}, i.e, the occupancy of $s_j$ is decreased. As a result the penalty transfer edges leading out of $s_j$ decrease in weight. As can be expected, we get a similar set of cases and the details can be skipped. 

\begin{lemma}
Let $ \mathcal{A}1$ be an allotment   that has no negative cost loops. 
Say that we modify $\mathcal{A}1$  by
 identifying a distinguished vertex $s_j$ and modifying weights of edges in $\Gamma_{pen}(\mathcal{A}1)$,  that are directed from $s_j$ to all other nodes.
Let $ \mathcal{A}2$ be the resulting allotment, and let $C_{min}$ be the loop of least cost induced by $\mathcal{A}2$ that begins at $s_j$ and contains an   edge of modified weight. We have the following:\\
(a) If $C_{min}$ has cost greater than or equal to zero, then $ \mathcal{A}2$ induces no negative loops.\\
(b) Say that $C_{min}$ has cost less than zero, and that modifying $ \mathcal{A}2$ to remove $C_{min}$   results in the allotment  $ \mathcal{A}3$. $ \mathcal{A}3$  does not induce any  negative loop.
\label{mainlemma3}
\end{lemma} 

  \subsection{Changing the Capacity or Penalty for a Service Center}
  Using the above methods, the capacity of   service center $s_j$, or the penalty function associated with $s_j$  can be modified. In this case the requirement for confining the changes to one demand node, is met by restricting the 
  capacity increase (or decrease) to one demand unit.   
  If we are modifying the penalty function, the modification is restricted to shifting the function  by one unit along the 
 X-axis. In case of an increase in capacity or a right shift of the penalty function, the cost of some edges leading out of $s_j$ in $\Gamma_{pen}()$ is decreased. This corresponds to the possibility that the cost of the allotment can be lowered by moving a demand node  into the distinguished service center. The optimality of the resulting solution follows from Lemmma \ref{mainlemma3}. In case of a capacity decrement or a left shift of the penalty function, the resulting case is identical to what we see in \ref{mainlemma2}.

\section{Conclusion}
\label{con}
Load Balanced Demand Distribution (LBDD) is a societally important problem. LBDD problem can be reduced to min-cost bipartite matching problem. Here we have proposed a new approach that reduces the running time from 
$O(n^3k)$ to $O(n(k^3+k^2 \log n ))$

\bibliographystyle{abbrv}
\bibliography{llncs_lbnvd}

\appendix

\end{document}